\documentclass[british]{article}
\usepackage[T1]{fontenc}
\usepackage[utf8]{inputenc}
\usepackage{xcolor}
\usepackage{textcomp}
\usepackage{amsmath}
\usepackage{amsthm}
\usepackage{amssymb}
\usepackage[authoryear]{natbib}

\makeatletter
\newcommand{\lyxaddress}[1]{
	\par {\raggedright #1
	\vspace{1.4em}
	\noindent\par}
}
\theoremstyle{plain}
\newtheorem{thm}{\protect\theoremname}
\theoremstyle{plain}
\newtheorem{cor}[thm]{\protect\corollaryname}
\theoremstyle{plain}
\newtheorem{lem}[thm]{\protect\lemmaname}

\usepackage{tikz}
\usetikzlibrary{shapes,snakes}
\usepackage{pgfplots}
\usepackage{hyperref}
\usepackage{breakurl}
\usepackage{latexsym}
\usepackage{stmaryrd}

\makeatother

\usepackage{babel}
\providecommand{\corollaryname}{Corollary}
\providecommand{\lemmaname}{Lemma}
\providecommand{\theoremname}{Theorem}

\begin{document}
\title{Relational electromagnetism and the electromagnetic force}
\author{M.A. Natiello$^{\dagger}$and H.G. Solari$^{\ddagger}$}
\maketitle

\lyxaddress{$^{\dagger}$Centre for Mathematical Sciences, Lund University;}

\lyxaddress{$^{\ddagger}$Departamento de Física, FCEN-UBA and IFIBA-CONICET}
\begin{abstract}
The force exerted by an electromagnetic body on another body in relative
motion, and its minimal expression, the force on moving charges or
\emph{Lorentz' force} constitute the link between electromagnetism
and mechanics. Expressions for the force were produced first by Maxwell
and later by H. A. Lorentz, but their expressions disagree. The construction
process was the result, in both cases, of analogies rooted in the
idea of the ether. Yet, the expression of the force has remained despite
its production context. We present a path to the electromagnetic force
that starts from Ludwig Lorenz' relational electromagnetism which
was available at the time of Lorentz' work. The present mathematical
abduction does not rest on analogies. Following this path we show
that relational electromagnetism, as pursued by the Göttingen school,
is consistent with Maxwell's transformation laws and compatible with
the idea that the ``speed of light'' takes the same value in all
(inertial) frames of reference, while it cannot be conceived on the
basis of analogies with material motion.

\textbf{ORCID numbers:}M. A. Natiello: 0000-0002-9481-7454 H. G. Solari:
0000-0003-4287-1878
\end{abstract}

\section{Introduction and Historical Background}

By the mid XIX-th century two cultures, or rather two forms of conceiving
science, came into conflict around the issues posed by electromagnetic
phenomena. We will label them the Göttingen and Berlin schools. Franz
Neumann as well as his son Carl, Weber, Riemann and Gauss belong to
the Göttingen school and support action-at-a-distance. Hertz, Heaviside,
Clausius and others represent the Berlin school, who denies action-at-a-distance
and relies on the \emph{ether} as a form of electromagnetic medium
sustaining the electromagnetic interaction. Maxwell has an intermediate
position. He acknowledges the achievements of the Göttingen school,
while he cannot conceive electromagnetic interactions without a medium
supporting the interaction while ``travelling'' from source to detector.
Indeed, he ends the Treatise by commenting on the Göttingen school,
finally stating \citep[{[866],}][]{maxw73} (our highlight):
\begin{quote}
But in all of these theories the question naturally occurs: If something
is transmitted from one particle to another at a distance, what is
its condition after it has left the one particle and before it has
reached the other? If this something is the potential energy of the
two particles, as in Neumann's theory, how are we to conceive this
energy as existing in a point of space, coinciding neither with the
one particle nor with the other? In fact, whenever energy is transmitted
from one body to another in time, there must be a medium or substance
in which the energy exists after it leaves one body and before it
reaches the other, for energy, as Torricelli remarked, 'is a quintessence
of so subtile a nature that it cannot be contained in any vessel except
the inmost substance of material things.' \textbf{Hence all these
theories lead to the conception of a medium in which the propagation
takes place}, and if we admit this medium as an hypothesis, I think
it ought to occupy a prominent place in our investigations, and that
we ought to endeavour to construct a mental representation of all
the details of its action, and this has been my constant aim in this
treatise.
\end{quote}
The paragraph shows Maxwell's urge to represent interactions in terms
of analogies with matter. Since bodies have a place in (subjective)
space, the paragraph assumes, without saying it, that we should consider
interactions (an abstract concept) in terms of the more accessible
(intuitive) idea of bodies. Because bodies, not interactions, have
a place in space.

On the other hand, the Berlin school resents of Maxwell's construction
as being ``impure'' noticing that his formulae are not always the
result of elaborations in terms of the ether. The Berlin school linked
understanding with images, mental representations of the observed,
the \emph{bild} theory \citep{dago04} after Helmholtz and Hertz,
whereas the Göttingen school aimed to provide a mathematical organisation
to electromagnetic knowledge where equations and their abduction/conceptualisation/construction
constitute an unity. \footnote{It is worth to keep in mind that non-scientific, irrational, factors
may have played a role as well since by 1868 Göttingen (Hannover)
was annexed to Prussia (hosting the school of Berlin). Clausius, a
Prussian nationalist and scientist, promoted at the same time an attack
on the scientific position on electromagnetism held by the Göttingen
school, from the point of view of the ether and of Berlin's epistemic
approach (See \citep{clau69} and \citep{arch86,NatielloManuscript-NATTCO-21}).} Another remarkable characteristic of the ether culture was the rejection
of the mathematical structures of mechanics (in particular action-at-a-distance).
In his posthumous work, Hertz attempted to reformulate Mechanics \citep{hertz89}.
He begins the preface by fostering doubts on Newton's mechanics in
relation to the ether:
\begin{quote}
All physicists agree that the problem of physics consists in tracing
the phenomena of nature back to the simple laws of mechanics. But
there is not the same agreement as to what these simple laws are {[}...{]}
But we have here no certainty as to what is simple and permissible,
and what is not: it is just here that we no longer find any general
agreement {[}...{]} So, for example, it is premature to attempt to
base the equations of motion of the ether upon the laws of mechanics
until we have obtained a perfect agreement as to what is understood
by this name.
\end{quote}
Lorentz' place in the development of electromagnetism is at a turning
point of it. Lorentz was fluent in both epistemic traditions and his
re-elaboration of Maxwell tries to keep a balance on the different
views, while committed to the ether interpretation. However, these
two views cannot coexist and Lorentz' approach turned to be incompatible
with both.

The history of what we know today as Lorentz' force has been presented
in \citep[Appendix A,][]{assi94}. It appears in Maxwell's treatise
as:
\[
F=qE+J\times B
\]
where $J=j+\epsilon_{0}\frac{\partial E}{\partial t}$ (written in
modern notation). The term associated to $j$ was further decomposed
in contributions of galvanic currents and that of charges in relative
motion and it was inferred after considerations of energy associate
to the interaction of two circuits. The term $\epsilon_{0}\frac{\partial E}{\partial t}\times B$
was added by analogy under the idea that the displacement current
was a current associated to the ether. Thompson \citep{thomson1881}
obtained the force following Maxwell's method and some special considerations,
Heaviside \citep{heaviside1889} later obtained the expression introducing
some changes in the argumentation but basically along the same lines.

Lorentz deduction of the electromagnetic force on a charged particle
from Maxwell's background is presented in \citep[§74-§80,][]{lorentz1892CorpsMouvants}.
Lorentz attempts to reintroduce mathematical structures but his argumentation
is marred by the use of steps founded only in the idea of the ether.
Lorentz assumed a simultaneous displacement of a body with respect
to the ether and to the complement of electromagnetic bodies. As soon
as the argumentation rests on the ether, it falls apart when the ether
is suppressed, unless there is an alternative for producing the result
resting on the complement of electromagnetic bodies.

The positive side of both Maxwell's and Lorentz' attempts is that
they derive the mechanical force due to electromagnetic phenomena
as a basic part of the electromagnetic theory. The theory in this
way is connected (unified) with mechanics from within itself. Indeed,
Lorentz writes down the five ``fundamental equations'' of electrodynamics
(Maxwell's four and the mechanical force) \citep{lorentz1899simplified}.
The downside is that both Maxwell (through the displacement current
and the so-called Ampère-Maxwell law) and Lorentz include the ether
in their derivations. The ether has subsequently been banned from
electrodynamics, leaving us with the uncomfortable situation of having
a theory rejecting the ether while its fundamental pillars at the
same time rest on the ether.

In the classical textbooks of the XXth. century \citep{panofsky1955classical,jackson1962classical,feynman1965,purc65},
the Lorentz force is presented as complementing Maxwell's equations
and its ultimate support is left to its success --after using it
intuitively-- as well as to its consistence with Lorentz transformations.
There is no derivation of the force from observations or first principles
in the cited books. Indeed, some of the mentioned books present the
Lorentz force as an independent posit, others refer to the original
work of Lorentz, and finally other books \citep{sope75classical,rohr07classical,zang12modern}
``reverse-engineer'' a lagrangian, formally deriving the force from
that lagrangian. It must be observed that the ad-hoc lagrangian is
not the one in Lorentz' while the mathematical operations assume independence
of currents and fields, only to arrive to the conclusion that they
are not independent, this is, showing that the approach is inconsistent.
From a historical point of view, Faraday writes, in the manuscript
entitled “Thoughts on ray vibrations“ \citep[p. 447]{fara55}:
\begin{quote}
The point intended to be set forth for consideration of the hearers
was, whether it was not possible that the vibrations that in certain
theory are assumed to account for radiation and radiant phænomena
cannot occur in the lines of force which connect particles, and consequently
masses of matter together; a motion that, as far as it is admitted,
will dispense with the æther, which in another view is supposed to
be the medium in which this vibrations take place.

You are aware of the speculation which I some times since uttered
respecting that view of the nature of matter that considers its ultimate
atoms as centres of force, and not as so many little bodies surrounded
by forces, the bodies being considered in the abstract as independent
of the forces and capable of existing without them. In the later view,
this little atoms have a definite form and a certain limited size;
in the former view such is not the case, for that which represents
size may be considered as extending to any distance to which the lines
of force of the particle extends: the particle indeed is supposed
to exist only by these forces, and where they are it is.
\end{quote}
This influential work is cited by Maxwell \citep{maxw64} in his ground-breaking
work. The proper conception of Maxwell's electromagnetism rests upon
the idea that force-matter or field-matter constitute dualities, hence
it is not possible to conceive a change in matter (or its state of
motion) and a change of its action (real or potential) upon other
bodies independently. Maxwell's equations have this duality in their
foundation.

The main goal of the present work is to obtain the expression of an
electromagnetic force without any reference to the ether, using the
theoretical tools available by 1873 (Maxwell's Treatise). Our claim
is that at that historical point it was possible to formulate the
electromagnetic force on charged particles (and Maxwell equations)
completely ether-free, resorting to research that was known around
1892.

We will show how the line of thought that Lorentz followed can be
implemented taking inspiration in Maxwell's approach without resorting
to the ether, thus moving away from absolute space and back to a relational
space. In so doing, we need to begin by reconsidering Maxwell's electromagnetism
from the standpoint of Ludwig Lorenz, who was the first in deducing
from experiments the wave equations for light as a transversal wave
phenomena \citep{lore61}. After removing the ether, Lorenz argument
\citep{lore67} must be acknowledged as the first \footnote{In his Treatise, Maxwell reclaims the primacy in time of his work
over Lorenz' while he does not contend Lorenz' rational argument \citep[{[805],}][]{maxw73}.} correct conception of electromagnetism.

After his experimental observations \citep{lore61}, Ludwig Lorenz
proposed that light is a transverse vibration phenomenon responding
to the wave equation. The idea of electromagnetic waves was circulating
at least since 1857 after work by Kirchhoff \citet{kirchhoff1857liv}
and Weber \citep{poge57,webe64}. Their waves, however, were present
inside material media (e.g., in wires). Lorenz' work \citep{lore67}
combines elements from both Weber and Kirchhoff, regarding light as
the perception of an electromagnetic oscillation outside matter. Lorenz
considered that Weber's equations had experimental support only in
experiments involving slow motion and could be regarded as a first
order approximation to wave propagation\footnote{Notice that by those times Weber's theory was able to explain all
electromagnetic phenomena experimentally known, as Maxwell recognised
in his 1865 foundational work.}.

With a philosophical position well aligned with that of Faraday who
wrote
\begin{quote}
But it is always safe and philosophic to distinguish, as much as is
in our power, fact from theory; the experience of past ages is sufficient
to show us the wisdom of such a course; and considering the constant
tendency of the mind to rest on an assumption, and, when it answers
every present purpose, to forget that it is an assumption, we ought
to remember that it, in such cases, becomes a prejudice, and inevitably
interferes, more or less, with a clear-sighted judgement. \citep[p.285]{fara44}
\end{quote}
Lorenz self-restrained from formulating physical hypothesis and considered
the matter (perhaps in more practical terms)
\begin{quote}
Hence it would probably be best to admit that in the present state
of science we can form no conception of physical reason of forces
and of their working in the interior of bodies; and therefore (at
present, at all events) we must choose another way, free from all
physical hypothesis, in order, if possible, to develop theory step
by step in such a manner that further progress of a future time will
not nullify the results obtained.
\end{quote}
The duality between matter and action, central to the thought of Faraday
in his vibrating rays theory \citep[p. 447]{fara55}, was put in more
mathematical terms by Lorenz: light corresponded with the electrical
activity in matter and its propagation at distance was subject to
matching conditions with the electrical activity inside matter. This
is the meaning of the ``identity of the vibrations of light with
electrical currents'' announced in the title of his contribution.

We will show in the next Section that Lorenz' approach provides an
alternative to the standard argument introducing Maxwell's displacement
current in the Ampère-Maxwell law (see eq. \ref{eq:displacement}
below).

We will attempt a reconstruction of electromagnetism starting from
Lorenz standpoint, recovering the energy formula by Maxwell and further
discussing Lorentz Lagrangian as an alternative starting point to
Lorenz. Next, we will address the Lorentz force in this setting showing
how it belongs to it. We will further show that, consistent with this
deduction, there is an alternative form for the force. The transformation
from one form to the other was envisioned by Maxwell \citep[{[600-601],}][]{maxw73}.
We give further evidence that the current view of electromagnetism
is more the consequence of the epistemological shift towards pragmatism
than to experimental facts. The actual use of electromagnetism and
its confrontation with experiments corresponds to the relational view,
while the constructive inconsistencies of the Lorentz force correspond
to its relation to the ether as an absolute reference.

\section{The Lorentz Force}

What came to be known as the Lorentz force is the outcome of the derivation
in \citep[§74-§80,][]{lorentz1892CorpsMouvants} for the force exerted
on a distinguished body (the probe) by the ether. The action of the
ether is assumed to be conveyed by the fields $E$ and $B$, which
satisfy Maxwell equations according to Heaviside, i.e., regarded from
a reference system at rest with the ether. Lorentz takes the expressions
for kinetic and potential electromagnetic energies from Maxwell, combining
them in Lagrangian form to obtain (in §77) $\nabla\times E=-\partial B/\partial t$,
the ``fourth'' Maxwell equation (not explicitly stated by Maxwell).
Subsequently, in §80 the Lorentz force is extracted from the Lagrangian,
for a probe taken to be a rigid solid, consisting of particles with
some charge density (nonzero only at the location of each particle)
that adds up to a smooth charge density $\rho$. The procedure in
both §77 and §80 is not purely mathematical but it responds to the
assumption on the ether, to which fields and velocities refer. Three
main differences with Maxwell are that (a) Maxwell never wrote down
a Lagrangian but worked with the energy contributions instead, (b)
Maxwell considers virtual displacements of the secondary circuit while
Lorentz considers that a virtual displacement corresponds at the same
time (and the same amount) to a displacement relative to the electromagnetic
bodies producing the fields and with respect to the ether (in §76).
One can hardly imagine something different than the ether being the
reference space associated to the sources of the field. Finally, (c),
the velocity of the charge in motion is for Lorentz the velocity with
respect to the ether, while Maxwell's velocity in {[}598{]} could
be interpreted either as relative to the primary circuit (the source
of electromagnetic fields) --as in Faraday's original concept--
or relative to the ether. The current use of the force adheres to
Faraday's view, thus becoming inconsistent with the derivation, inasmuch
the latter rests on the ether.

Before we proceed to a derivation of the force, it is worth to consider
the difference between Lorentz' and Maxwell's expressions for the
force, since Lorentz appears to be following Maxwell steps up to some
point. Maxwell restrained from making an early decision on the nature
of electricity. He considered a piece of matter carrying a current
and included the contribution of the ``displacement current'' just
by analogy with galvanic currents. The resulting force was determined
up to a contribution consisting of the gradient of a potential. Lorentz,
in turn, adopted the idea of material corpuscles (Weber's hypothesis
\citep{webe46}). He considered that no galvanic current was present
in the moving body and introduced the \emph{Lorentz current }which
is no more than the description of a moving (charged) body with respect
to the reference frame of the ether, made in classical terms. For
both of them, Maxwell and Lorentz, the body was a rigid solid and
the force was to be determined by a variational method. Lorentz acknowledges
the influence of the Göttingen school, when he claims: ``The influence
that was suffered by a particle B due to the vicinity of a second
one A, indeed depends on the motion of the latter, but not on its
instantaneous motion. Much more relevant is the motion of A some time
earlier, and the adopted law corresponds to the requirement for the
theory of electrodynamics, that was presented by Gauss in 1845 in
his known letter to Weber \citep[bd.5 p. 627-629,][]{gaus70}''.
Nevertheless, the extent of this influence is difficult to gauge,
since his derivation of the electromagnetic force rests on different
premises. Maxwell, on the other hand, acknowledges that Ludwig Lorenz'
approach leads to the same electromagnetic equations but he does not
explore contact points and differences any further.

\subsection{A word about the ether and absolute space.}

Absolute space had earned a bad reputation by the beginning of the
XIX century, only the relational space appeared to matter for scientists.
The difficulties of conceiving the wave propagation of electromagnetic
phenomena as well as the philosophical belief that ``matter acts
where it is'' \footnote{Reverted brilliantly by Faraday into \emph{matter is where it acts},
indicating that matter can only be inferred but not sensed, thus matter
is a belief while action is real.} emerged as the concept of an electromagnetic ether (having different
and sometimes contradictory properties for different authors). Lorentz
shortly addressed this problem:
\begin{quote}
That we cannot speak about an \emph{absolute} rest of the aether,
is self-evident; this expression would not even make sense ... When
I say for the sake of brevity, that the aether would be at rest, then
this only means that one part of this medium does not move against
the other one and that all perceptible motions are relative motions
of the celestial bodies in relation to the aether. \citep[p. 1,][]{lorentz1895attempt}
\end{quote}
For Lorentz, the ether had to be material but transparent to the (ponderable)
bodies. He did not fully convince other authors such as \citep{einstein1907relativity,ritz1908recherches}
for whom Lorentz' ether was just absolute space.

\subsection{Ether-free electrodynamics}

\subsubsection{Maxwell's equations}

The starting point of our task are Maxwell equations \citep{maxw73}
that in modern notation, with $C^{2}=\left(\mu_{0}\epsilon_{0}\right)^{-1}$,
can be stated as
\begin{eqnarray}
B & = & \nabla\times A\label{eq:magneticfield}\\
E & = & -\frac{\partial A}{\partial t}-\nabla V\label{eq:electricfield}\\
\epsilon_{0}\nabla\cdot E & = & \rho\label{eq:charges}\\
\mu_{0}j+\frac{1}{C^{2}}\frac{\partial E}{\partial t} & = & \nabla\times B\label{eq:displacement}
\end{eqnarray}
All quantities are evaluated at a reference point, $x$, relative
to the reference frame in which the distribution of charge is given
by $\rho(x,t)$ and satisfies the continuity equation ${\displaystyle \frac{\partial\rho(x,t)}{\partial t}}+\nabla\cdot j(x,t)=0$
\citep[{[295],}][]{maxw73}. In other words, it is the frame where
the fields $E,B$ have been determined through $\rho$ and $j$. The
equations correspond to the electromagnetic momentum, $A,$ (today
called vector potential), the magnetic induction, $B$ \citep[{[619] eq. A,}][]{maxw73},
the electric field intensity $E$ \citep[{[619], eq. B,}][]{maxw73}
(being the sum of an induction contribution \citep[{[598],}][]{maxw73}
and a gradient generalising the electrostatic potential) and the galvanic
current, $j$. The third equation, taken from \citep[{[612],}][]{maxw73}
and inspired in Coulomb and Faraday, generalises here Poisson's equation
from electrostatics, while the fourth equation \citep[{[619], eq. E,}][]{maxw73}
refers to \citep[{[610], eq. H*,}][]{maxw73} \footnote{The same equations can be found in \citep{maxw64} Section III under
the equation labels: B, 35, G, C (taking into account A) and H.}.

Maxwell obtained eq. \ref{eq:displacement} through considerations
about the ether. He argued that the propagation of electromagnetic
waves asks for a propagation medium, by analogy with ``the flight
of a material substance through space'' \citep[{[866]}][]{maxw73}.
However, Ludwig Lorenz \citep{lore67} expresses discomfort with the
ether hypothesis, which had only been useful to ``furnish a basis
for our imagination'' (p. 287). Consequently, he offered an ether-free
derivation of the field equations that took into account Gauss' suggestion
of delayed action at a distance \citep[bd.5 p. 627-629,][]{gaus70}
\footnote{Other authors arrived to similar equations, also without involving
the ether \citep{riem67,bett67,neum68}. These authors (but not Lorenz)
were in turn criticised for not involving the ether in their conceptualisation
\citep{clau69}.}. For static situations, the electromagnetic potentials at a point
$x$ (and time $t$) originating in a charge and current distribution
over a domain in space described through the coordinate $y$ are defined
through Poissons's equation starting from charges and currents, i.e.,
\begin{equation}
(A,\frac{V}{C})(x,t)=\frac{\mu_{0}}{4\pi}\int\left(\frac{(j,\rho C)(y,t)}{|x-y|}\right)d^{3}y.\label{eq:Kernel}
\end{equation}
For the vector potential $A$ this idea was originally advanced by
Franz Neumann \citep{neum1846-induction}.

Inspired in Gauss' proposal, Lorenz considered that the charges and
currents contributing to the potentials at time $t$ act with delay,
after being originated in a previous time $s=t-{\displaystyle \frac{|x-y|}{C}}$:

\begin{equation}
(A,\frac{V}{C})(x,t)=\frac{\mu_{0}}{4\pi}\int\left(\frac{(j,\rho C)(y,t-\frac{1}{C}|x-y|)}{|x-y|}\right)d^{3}y,\label{eq:Lorenz}
\end{equation}
where only galvanic currents and actual charges are involved. After
standard operations of vector calculus (see Lemma \ref{lem:wave},
Appendix \ref{sec:Proofs}) it is found that 
\begin{eqnarray}
\Box A & = & -\mu_{0}j\label{eq:propagation}\\
\Box V & = & -\frac{1}{\epsilon_{0}}\rho\nonumber 
\end{eqnarray}
where $\Box={\displaystyle \Delta-\frac{\partial^{2}}{\partial t^{2}}}$
is D'Alembert's wave operator. In Lemma \ref{lem:eq4-no-ether}, Appendix
\ref{sec:Proofs} we derive eq. \ref{eq:displacement} from Lorenz'
setup, taking advantage of the \emph{Lorenz gauge}. We note on passing
that eq. \ref{eq:Kernel} is not the only possible definition of $A$
and $V$ that is compatible with the wave equation.

\subsubsection{Charge and current distribution of electrified bodies in motion in
a relational framework}

Following Lorentz, charges are assumed to be rigid bodies \citep[§75c,][]{lorentz1892CorpsMouvants}
or ``corpuscles'' and all charges are considered to be accounted
for directly, so there is no need to introduce polarisation/induction
fields. With this setup, the force on a probe charge will be identical
to the electromotive force.

While the idea that a charge in movement corresponds to a current
goes back (at least) to Weber \citep{webe46}, the form of constructing
the current had been the subject of controversy. For Maxwell the velocities
that matter were relative velocities between primary and secondary
circuits in full agreement with Faraday \citet[{[568-583]}][]{maxw73}.
For the sake of the argument we consider both circuits as rigid solids.
Lorentz in turn considered that the velocities had to be considered
relative to the ether, an idea that has to be disregarded since no
evidence of the ether can be found. Hence, only the Faraday-Weber-Maxwell
(and others) idea of \emph{relative velocity} appears to be tenable.
We let $\dot{\bar{x}}(t)$ be the relative velocity between two reference
points in the primary and secondary bodies. We hereby define current
according to

\begin{eqnarray}
(x,t) & = & (y+\bar{x}(t),t)\nonumber \\
\bar{\rho}_{2}(x,t) & = & \rho_{2}(x-\bar{x}(t),t)\nonumber \\
\bar{j}_{2}(x,t) & = & j_{2}(x-\bar{x}(t),t)+\dot{\bar{x}}\rho_{2}(x-\bar{x}(t),t)\label{eq:Current}
\end{eqnarray}
where $y$ is a ``local'' coordinate of the secondary body (labelled
with the index $2$, the body where we aim to compute the acting force)
and $x$ is the coordinate relative to the reference frame above in
which the potentials are $A$ and $V$. Hence, neither ether nor any
inertial frame is considered, but only the relative coordinate between
the electromagnetic bodies undergoing mutual interaction. The charge
density and internal current $(\rho_{2}(y,t),j_{2}(y,t))$ are described
in a frame at rest with respect to the rigid body $2$, while $(\bar{\rho}_{2}(x,t),\bar{j}_{2}(x,t))$
are their expressions with reference to the primary circuit. Also,
it is assumed that the body does not rotate. In either reference frame,
charge and current for body $2$ also satisfy the continuity equation.

This definition of current (eq. \ref{eq:Current}) reduces to Lorentz'
current for the particular case $j_{2}=0$, only that the velocity
is now relational, while the first term in eq. \ref{eq:Current} allows
for other sources of current, as it was entertained by Faraday and
Maxwell.

\subsubsection{Maxwell's energy and Lorentz' Lagrangian revisited}

Maxwell's energy is introduced through a process in which matter acquires
its electromagnetic state \citep{maxw73}. The electrostatic energy
is obtained in {[}630-631{]} bringing charges from infinity. Next,
the magnetostatic energy is obtained in similar form in {[}632-633{]},
based upon magnetostatic results previously obtained in {[}389{]}.
Maxwell proceeds to add an electrokinetic energy due to the currents
{[}634-635{]}. In the even numbered articles he presents the physical
idea and in the odd numbered articles he transforms the expression
using integration by parts.

Thus, the energy required to create a given electromagnetic state
(a distribution of charges and currents) can be regarded as the time-integral
of the power, first bringing charges from a condition of zero energy
(from ``infinity'') working against $-\nabla V$ and bringing also
the current distribution, now working against the electromagnetic
momentum $A$: 
\[
{\cal P}=\left(\nabla V+\frac{\partial A}{\partial t}\right)\cdot j=-E\cdot j
\]
the time-integral from a situation in which ${\cal E}(0)=0$, leads
to 
\begin{equation}
{\cal E}=-\int_{0}^{t}dt\int d^{3}x\left(E\cdot j\right)=\int_{0}^{t}dt\int d^{3}x\left(\left(\nabla V+\frac{\partial A}{\partial t}\right)\cdot j\right).\label{eq:power}
\end{equation}
Assuming that all of $|B|^{2},|E|^{2},A\cdot j,V\rho$ decrease faster
than $\frac{1}{r^{2}}$ at infinity (a hypothesis needed for most
manipulations performed by Maxwell and Lorentz) and applying Gauss
theorem to convert volume integrals of a divergence into surface integrals
(vanishing at infinity by the assumption), a straightforward computation
(Lemma \ref{lem:energy-final}, Appendix \ref{sec:Proofs}) shows
that the energy provided by the electrification of the body is
\begin{eqnarray}
{\cal E} & = & \frac{1}{2}\int d^{3}x\left(\frac{1}{\mu_{0}}|B|^{2}+\epsilon_{0}|E|^{2}\right).\label{eq:energy}
\end{eqnarray}
From the electrostatic and magetostatic situations it is clear that
the individual terms correspond to the electrokinetic, $T=\frac{1}{2}\int d^{3}x\left(\frac{1}{\mu_{0}}|B|^{2}\right)$,
and potential energies, $U=\frac{1}{2}\int d^{3}x\left(\epsilon_{0}|E|^{2}\right)$,
thus suggesting that the electromagnetic Lagrangian reads ${\cal L}=T-U$
and the action ${\cal A}=\int dt\,{\cal L}$. Another straightforward
computation and application of Gauss theorem (Lemma \ref{lem:action-final},
Appendix \ref{sec:Proofs}) leads to 
\begin{eqnarray}
{\cal {\cal A}} & = & \frac{1}{2}\int dt\,\int\left(\frac{1}{\mu_{0}}|B|^{2}-\epsilon_{0}|E|^{2}\right)\,d^{3}x\nonumber \\
 & = & \frac{1}{2}\int dt\,\int\left(A\cdot j-\rho V\right)\,d^{3}x+f(t)-f(t_{0})\label{eq:action}
\end{eqnarray}

The dynamic equations are independent of $f$, so we can proceed from
here by taking $f\equiv0$. While the intuition of terms can be taken
as a suggestion, we must ask: What kind of dynamical situation is
reflected by this action?
\begin{thm}
Let $(\mathsf{A},\mathsf{V})$ be the known values of the electromagnetic
potentials in a piece of matter supported on a region of space with
characteristic function $\chi$. Then, under the assumptions \footnote{The continuity equation, eq.(\ref{eq:Lorenz}) and ${\displaystyle \nabla\cdot A+{\displaystyle \frac{1}{C^{2}}\frac{\partial V}{\partial t}}=0}$,
see Appendix.} of Lemma \ref{lem:GaugeLorenz}, Hamilton's principle of least action
\citep[Ch 3, 13 A p. 59, ][]{arno89}, $\delta{\cal A}=0$, subject
to the constraints given by $(\mathsf{A},\mathsf{V})$ implies that
the manifestation of the potentials outside matter obeys the wave
equation.
\end{thm}

\begin{proof}
The result follows from the computation of the minimal action under
\begin{eqnarray*}
(V-\mathsf{V})\chi & = & 0\\
(A-\mathsf{A})\chi & = & 0
\end{eqnarray*}
Multiplying the constraints by the Lagrange multipliers $\lambda$
and $\kappa$ (the latter a vector), while we use the shorthand notations
$B=\nabla\times A$ and $E=\left(-\frac{\partial A}{\partial t}-\nabla V\right)$,
we need to variate
\[
{\cal A}=\frac{1}{2}\int dt\left(\int\left(\frac{1}{\mu_{0}}|B|^{2}-\epsilon_{0}|E|^{2}-\kappa\cdot(A-\mathsf{A})\chi+\lambda(V-\mathsf{V})\chi\right)\,d^{3}x\right).
\]
 After variation and applying Gauss theorem as usual leads to
\begin{eqnarray*}
\Box A & = & -\chi\kappa\\
\Box V & = & -\chi\lambda
\end{eqnarray*}
which allows us to identify $\mu_{0}j=\chi\kappa$ (the density of
current inside the material responsible for $\mathsf{A}$ and ${\displaystyle \frac{\rho}{\epsilon_{0}}=}\chi\lambda$
the density of charge responsible for $\mathsf{V}$.
\end{proof}
This theorem deserves to be called Lorenz' theorem since he wrote
about these relations \citep[p.300,][]{lore67}: ``This result is
a new proof of the identity of the vibrations of light with electrical
currents; for it is clear now that not only the laws of light can
be deduced from those of electrical currents, but that the converse
way may be pursued, provided the same limiting conditions are added
which the theory of light requires.''

Thus, electromagnetism has been summarised in terms of the potential
$V$ and vector potential (or electrokinetic momentum, according to
Maxwell-Faraday) $A$ (defined via eq. \ref{eq:Lorenz}), the principle
of least action and the continuity equation. All other relations consist
in either naming intermediate quantities, such as $B=\nabla\times A$
and $E=-(\frac{\partial A}{\partial t}+\nabla V)$, or stating consequences
of these definitions, such as $\frac{\partial}{\partial t}\nabla\times A=\nabla\times\frac{\partial A}{\partial t}$
(which leads immediately to Faraday's induction law), or $\nabla\cdot E=\rho$
(that follows after the use of the continuity equation).

The connection of this theoretical structure with observations in
nature is given by the phenomenological map \citep{Solari21-phenomenologico}
providing observational content to charges and currents in matter.
We are now ready to show how Lorentz' force belongs to this context
and is already inscribed in the former equations.

\subsubsection{Deduction of Lorentz' force revisited}

When we consider two pieces of electrified matter in interaction we
can envisage a different form of constructing the system. In the first
step, the bodies are far apart, so that we can assume that they do
not interact, and are electrified to reach their actual state. Next,
they will be brought together to their corresponding mechanical positions
in terms of a thought process called a \emph{virtual displacement}.
The formalisation of this idea already present in Maxwell is called
a \emph{virtual variation} \citep[Ch 4, 21 B p. 92, ][]{arno89}.
The force associated to this virtual displacement will then result
from the variation of the interaction terms in the Lagrangian. Given
the electromagnetic contribution to the action, determining the contribution
to the force amounts to applying Hamilton's principle using a virtual
displacement of the probe (which we indicate with subindex 2) with
respect to the primary circuit producing the fields (subindex 1).
In formulae, we must request
\[
\delta_{\bar{x}(t)}{\cal A}=0
\]
with $\delta\bar{x}(t_{0})=0$ and $\delta\bar{x}(t)=0$, where $\bar{x}(t)$
denotes the relative distance between probe and primary circuit. In
so doing, we must take into account that there is a corresponding
variation of the velocity $\delta\dot{\bar{x}}(t)$. Recall that according
to \ref{eq:propagation} and \ref{eq:Current}, $\bar{x}$ occurs
in $(\rho_{2},j_{2})$ and $\dot{\bar{x}}$ occurs only in $j_{2}$.
We state the result as a theorem:
\begin{thm}
\label{thm:TheForce} Assuming that all of $|B|^{2},|E|^{2},A,j,V,\rho$
decrease faster than $\frac{1}{r^{2}}$ at infinity and under the
conditions of Lemmas \ref{lem:wave}--\ref{lem:gauge-independent}
(in particular the action of eq. \ref{eq:action}) the electromagnetic
force \textup{
\[
F_{em}=\bar{j}_{2}\times B_{1}+\bar{\rho}_{2}E_{1}
\]
}on the probe can be deduced from Hamilton's principle of minimal
action (\textup{$\delta_{\bar{x}(t)}{\cal A}=0$}) using a virtual
displacement $\delta_{\bar{x}}$ of the probe (which we indicate with
subindex 2) with respect to the primary circuit producing the fields
(subindex 1).
\end{thm}

For the proof, see Subsection \ref{subsec:TheProof} , Appendix \ref{sec:Proofs}.
Lorentz considered only the case $j_{2}=0$, hence $\bar{j}_{2}=\bar{\rho}_{2}\dot{\bar{x}}$.
Maxwell considered the general case in \citet[{[602]}][]{maxw73}
but using his ``total current'' instead, consistent with his belief
in the ether-based displacement current as a material current. However,
his elaboration is based on galvanic currents, finally modifying the
final result by analogy \citet[{[619], eq. C}]{maxw73}. The present
approach differs with Lorentz' in the broader concept of current and
in that the participating quantities are fully relational and describe
the interaction between probe and primary circuit.

\subsection{Subjective perspective}

Since the Lagrangian is expressed by an integral, we are free to change
the integration variable by a fixed translation leaving the integral
unchanged. Actually, we can use a different translation for each time
in \eqref{eq:action}. We propose to change integration variable from
$x$ to $z,$ with $x=z+\bar{x}(t)$. Instead of performing the change
in eq.\eqref{eq:variation} we will save effort and perform it in
eq.\eqref{eq:variada}, prior to the partial integration in time,
namely 
\begin{eqnarray*}
\delta{\cal A} & = & \int dt\int d^{3}x\left[\bar{j}_{2}\cdot\left(\delta\bar{x}\cdot\nabla\right)A-\bar{\rho}_{2}\left(\delta\bar{x}\cdot\nabla\right)V+\delta\dot{\bar{x}}\cdot\left(A\bar{\rho}_{2}\right)\right]
\end{eqnarray*}
(with $\bar{j}_{2}(x,t),\quad\bar{\rho}_{2}(x,t)$ given by eq.\ref{eq:Current}).
We introduce the following notation

\begin{eqnarray}
z & = & x-\bar{x}(t)\nonumber \\
\underline{V}(z,t) & = & V(z+\bar{x}(t),t)\label{eq:underbar}\\
\underline{A}(z,t) & = & A(z+\bar{x}(t),t)\nonumber 
\end{eqnarray}

Hence, the variation reads now

\begin{eqnarray*}
\delta\int\!{\cal L}dt & = & \int dt\int d^{3}z\left[\left(j_{2}(z,t)+\dot{\bar{x}}\rho_{2}(z,t)\right)\cdot\left(\delta\bar{x}\cdot\nabla\right)\underline{A}-\rho_{2}(z,t)\left(\delta\bar{x}\cdot\nabla\right)\underline{V}\right]\\
 &  & +\int dt\int d^{3}z\left[\delta\dot{\bar{x}}\cdot\left(\underline{A}\rho_{2}(z,t)\right)\right]
\end{eqnarray*}
integrating by parts in time the last term and using the relations

\begin{eqnarray*}
\int dt\left[\delta\dot{\bar{x}}\cdot\left(\underline{A}\rho_{2}\right)\right] & = & \int dt\left[-\rho_{2}\delta\bar{x}\cdot\left[\frac{\partial\underline{A}}{\partial t}\right]-\delta\bar{x}\cdot\underline{A}\frac{\partial\rho_{2}}{\partial t}\right]\\
\left[\frac{\partial\underline{A}}{\partial t}\right] & \equiv & \frac{\partial}{\partial t}A(z+\bar{x}(t),t)=\left.\frac{\partial}{\partial t}A(x,t)\right|_{x=z+\bar{x}(t)}+\left(\dot{\bar{x}}\cdot\nabla\right)A(z+\bar{x}(t),t)\\
-\nabla\underline{V}-\left[\frac{\partial\underline{A}}{\partial t}\right] & = & -\nabla\underline{V}-\left.\frac{\partial}{\partial t}A(x,t)\right|_{x=z+\bar{x}(t)}-\left(\dot{\bar{x}}\cdot\nabla\right)\underline{A}
\end{eqnarray*}
(along with the continuity equation) we arrive after some algebra
and further use of Gauss' theorem to:

\begin{eqnarray*}
\delta\int\!{\cal L}dt & = & \int dt\int d^{3}z\left[j_{2}\cdot\left(\delta\bar{x}\cdot\nabla\right)\underline{A}-j_{2}\cdot\nabla\left(\delta\bar{x}\cdot\underline{A}\right)\right]\\
 &  & +\int dt\int d^{3}z\left[(\rho_{2}\dot{\bar{x}})\cdot\left(\delta\bar{x}\cdot\nabla\right)\underline{A}+\rho_{2}\delta\bar{x}\cdot(-\nabla\underline{V}-\frac{\partial\underline{A}}{\partial t})\right]
\end{eqnarray*}

Finally, the following relations (the second one valid for any sufficiently
differentiable scalar function $\Phi$)
\begin{eqnarray*}
\dot{x}\cdot\left(\delta\bar{x}\cdot\nabla\right)\underline{A} & = & \delta\bar{x}\cdot\nabla(\dot{x}\cdot\underline{A})\\
(\delta\bar{x}\cdot\nabla)\Phi(x,t) & = & -\delta\bar{x}\times(\nabla\times\Phi)+\nabla(\delta\bar{x}\cdot\Phi)\\
j_{2}\cdot\left(\delta\bar{x}\cdot\nabla\right)\underline{A}-j_{2}\cdot\nabla(\delta\bar{x}\cdot\underline{A}) & = & j_{2}\cdot\left(-\delta\bar{x}\times(\nabla\times\underline{A}\right)
\end{eqnarray*}
lead us to the next result:

\begin{eqnarray*}
\delta\int{\cal L}dt & = & \int dt\int d^{3}z\left[j_{2}\cdot\left(-\delta\bar{x}\times(\nabla\times\underline{A})\right)\right]\\
 &  & +\int dt\int d^{3}z\left[\delta\bar{x}\cdot\rho_{2}\left(-\left[\frac{\partial\underline{A}}{\partial t}\right]-\nabla\left(\underline{V}-\dot{\bar{x}}\cdot\underline{A}\right)\right)\right]\\
 & = & \int dt\int d^{3}z\,\delta\bar{x}\cdot\left[j_{2}\times\underline{B}+\rho_{2}\left(-{\displaystyle \left[\frac{\partial\underline{A}}{\partial t}\right]}-\nabla(\underline{V}-\dot{\bar{x}}\cdot\underline{A})\right)\right]
\end{eqnarray*}

Hence we have two expressions for the mechanical contribution of the
electromagnetic force: The one obtained from eq.\eqref{eq:variada}
above and the present one, i.e.,
\[
F_{em}=\!\int\!d^{3}x\left[\bar{j}_{2}\times B+\bar{\rho}_{2}E\right]=\!\int\!d^{3}z\left[j_{2}\times\underline{B}+\rho_{2}\left(-{\displaystyle \left[\frac{\partial\underline{A}}{\partial t}\right]}-\nabla(\underline{V}-\dot{\bar{x}}\cdot\underline{A})\right)\right]
\]
(recall the relation among functions defined in eqs. \ref{eq:Current}
and \ref{eq:underbar}).

We call this relation ``Maxwell's transformation theorem'', since
Maxwell showed the following result in \citet[{[601],}][]{maxw73}\footnote{Maxwell's result covers also rotating reference systems, but we skip
this case to keep the argumentation simple.}:
\begin{thm}
\textbf{(Maxwell's invariance theorem}) \label{Theorem-Maxwell's-invariance}:
Let $x=x^{\prime}+\bar{x}(t)$, and correspondingly $v=v^{\prime}+\dot{\bar{x}}$.
Define ${\displaystyle A^{\prime}(x^{\prime},t)\equiv A(x,t)}$ \footnote{Maxwell refers to this expression as: “the theory of the motion of
a body of invariable form“. For any property of matter, this relation
is immediate.}, then the value of the electromotive force at a point $x$ does not
depend on the choice of reference system if and only if $\psi(x,t)$
transforms as $\psi^{\prime}(x^{\prime},t)\equiv\psi(x^{\prime}-\bar{x},t)-\dot{\bar{x}}\cdot A^{\prime}(x^{\prime},t)$.
In formulae, $\mathcal{E}^{\prime}(x^{\prime},t)=\mathcal{E}(x,t)$,
where $\psi$ is an undetermined electrodynamic potential (introduced
for the sake of generality):
\end{thm}

\[
\mathcal{E}^{\prime}(x^{\prime},t)=v^{\prime}\times(\nabla\times A^{\prime}(x^{\prime},t))-\frac{\partial A^{\prime}(x^{\prime},t)}{\partial t}-\nabla\psi^{\prime}(x^{\prime},t).
\]

\begin{proof}
First, according to the definition, we have ${\displaystyle A^{\prime}(x^{\prime},t)=A(x^{\prime}+\bar{x},t)}$.

Next, we note that by straightforward vector calculus identities,
Maxwell's electromotive force (eq. B in \citep[{[598],}][]{maxw73})
can be restated as 
\[
\mathcal{E}(x,t)=-\frac{\partial A(x,t)}{\partial t}-\left(v\cdot\nabla\right)A(x,t)-\nabla\left(\psi(x,t)-v\cdot A(x,t)\right)
\]
In the new coordinate system we compute: 
\begin{eqnarray*}
\mathcal{E^{\prime}}(x^{\prime},t) & = & v^{\prime}\times(\nabla\times A^{\prime}(x^{\prime},t))-\frac{\partial A^{\prime}(x^{\prime},t)}{\partial t}-\nabla\psi^{\prime}(x^{\prime},t)\\
 & = & -\frac{\partial A^{\prime}}{\partial t}(x^{\prime},t)-\left(v^{\prime}\cdot\nabla\right)A^{\prime}-\nabla\left(\psi^{\prime}(x^{\prime},t)-v^{\prime}\cdot A^{\prime}(x^{\prime},t)\right).
\end{eqnarray*}
Subsequently, under the present assumption ${\displaystyle A^{\prime}(x^{\prime},t)\equiv A(x,t)}$
we may rewrite 
\[
\frac{\partial A^{\prime}(x^{\prime},t)}{\partial t}=\left.\frac{\partial A(x,t)}{\partial t}\right|_{x^{\prime}+\bar{x}}+\left(\dot{\bar{x}}\cdot\nabla\right)A
\]
leading to 
\begin{eqnarray*}
\mathcal{E^{\prime}}(x^{\prime},t) & = & -\frac{\partial A(x,t)}{\partial t}-\left(\dot{\bar{x}}\cdot\nabla\right)A-\left(v^{\prime}\cdot\nabla\right)A^{\prime}-\nabla\left(\psi^{\prime}(x^{\prime},t)-v^{\prime}\cdot A^{\prime}(x^{\prime},t)\right)\\
 & = & -\frac{\partial A(x,t)}{\partial t}-\left(v\cdot\nabla\right)A-\nabla\left(\psi^{\prime}(x^{\prime},t)-v^{\prime}\cdot A^{\prime}(x^{\prime},t)\right)\\
 & = & -\frac{\partial A(x,t)}{\partial t}-\left(v\cdot\nabla\right)A-\nabla\left(\psi(x,t)-\dot{\bar{x}}\cdot A(x,t)-v^{\prime}\cdot A^{\prime}(x^{\prime},t)\right)\\
 &  & -\frac{\partial A(x,t)}{\partial t}-\left(v\cdot\nabla\right)A-\nabla\left(\psi(x,t)-v\cdot A(x,t)\right)=\mathcal{E}(x,t)
\end{eqnarray*}
where to proceed from the second line to the third, we used the condition
$\psi^{\prime}(x^{\prime},t)\equiv\psi(x,t)-\dot{\bar{x}}\cdot A(x,t)$.
Note that $\psi$ is Maxwell's undetermined potential that once determined
in one system of reference transforms according to the theorem to
other systems.
\end{proof}
Maxwell's Art. {[}601{]} of the Treatise states “It appears from this
that the electromotive intensity is expressed by a formula of the
same type, whether the motions of the conductors be referred to fixed
axes or to axes moving in space, the only difference between the formulae
being that in the case of moving axes the electric potential $\psi$
must be changed into $\psi+\psi^{\prime}$.“

\subsubsection{No Arbitrariness Principle}

The invariance of the force as determined from the point of view of
the source (primary circuit) or of the target (probe, secondary circuit)
in the previous subsection follows from the general invariance of
integrals in front of coordinate changes, supported in Maxwell's invariance
theorem. The result is a special case of the No Arbitrariness Principle
(NAP) \citep{sola18b}, stating that the description of natural processes
cannot depend on arbitrary choices (in this case the choice of subjective
reference frame). Indeed, we can attain a fully relational description
of electromagnetic phenomena.

The consideration of three electrified bodies will allow us to inspect
this problem. Let $x_{ij}(t)$ the relative position between the $i$
and the $j$ body, with $i,j\in\left\{ 1,2,3\right\} $, clearly,
$x_{ii}=0$. The relative positions satisfy $x_{12}+x_{23}+x_{31}=0$
and correspondingly, the relative velocities satisfy $v_{12}+v_{23}+v_{31}=0$.
Let $\zeta_{i}^{i}$ be each of the components of the four dimensional
vector $(j_{i},\rho_{i})$ describing the density of currents and
the density of charges in body $i$, as perceived from a frame fixed
to itself, and $\zeta_{i}^{j}$ the same components as described from
the frame of body $j$. We will denote by ${\cal W}$ the operator
that produces the delayed propagation of the EM situation in the body
and, finally, $T(x)$ the operator that applies a time-dependent translation
to current+charge vector as in \ref{eq:Current}. Clearly $T(0)=Id$.
Then, given the potentials of body $i$ in its own frame, we get the
potentials in the body $j$ frame as
\[
(A_{i}^{j},V_{i}^{j})=\left({\cal W}T(x_{ji})\Box\right)(A_{i}^{i},V_{i}^{i})
\]
The operators $T(x_{ij}(t))$ form a group of transformations. It
is easy to verify that the product law is:
\[
T(x(t))T(y(t))=T((x+y)(t)
\]
Letting
\[
\tilde{T}(x_{ij})={\cal W}T(x_{ji})\Box
\]
 we notice that the operators $\tilde{T}(x_{ij})$ are conjugated
to $T(x_{ij})$ since the relation $\Box{\cal W}=Id$ and ${\cal W}\Box=Id$
produce the conjugation relation
\begin{eqnarray*}
\Box\tilde{T}(x_{ij}) & = & T(x_{ij})\Box\\
\tilde{T}(x_{ij}){\cal W} & = & {\cal W}T(x_{ij})
\end{eqnarray*}
Thus, the rigid translation of charge densities and currents $T(x_{ij})$
acts with a conjugate representation on the wave representatives.
This is the abstract content of Maxwell's theorem.

We can finally write all the fields in terms of one reference body,
say, the body with $i=1$ as
\[
(A_{1},V_{1})={\cal W}\left(\zeta_{1}^{1}+T(x_{12})\zeta_{2}^{2}+T(x_{13})\zeta_{3}^{3}\right)
\]
We next apply $\tilde{T}(x_{21})$ to this expression, we get
\begin{eqnarray*}
\tilde{T}(x_{21})(A_{1},V_{1}) & = & \tilde{T}(x_{21}){\cal W}\left(\zeta_{1}^{1}+T(x_{12})\zeta_{2}^{2}+T(x_{13})\zeta_{3}^{3}\right)\\
 & = & \left({\cal W}T(x_{21})\Box\right){\cal W}\left(\zeta_{1}^{1}+T(x_{12})\zeta_{2}^{2}+T(x_{13})\zeta_{3}^{3}\right)\\
 & = & \left({\cal W}T(x_{21})\right)\left(\zeta_{1}^{1}+T(x_{12})\zeta_{2}^{2}+T(x_{13})\zeta_{3}^{3}\right)\\
 & = & {\cal W}\left(T(x_{21})\zeta_{1}^{1}+T(x_{21})T(x_{12})\zeta_{2}^{2}+T(x_{21})T(x_{13})\zeta_{3}^{3}\right)\\
 & = & {\cal W}\left(T(x_{21})\zeta_{1}^{1}+\zeta_{2}^{2}+T(x_{23})\zeta_{3}^{3}\right)\\
 & = & (A_{2},V_{2})
\end{eqnarray*}
Which shows how the subjective representation of fields transforms
consistently with the time-dependent-translations group. When the
admitted reference is restricted to inertial bodies \citep{NatielloManuscript-NATTCO-21},
the group of transformations is the Galilean group.

Let us specify the above construction for the electromagnetic force
that bodies $2$ and $3$ exert on a moving charge $q_{1}$. Expressing
the fields as computed by bodies $2$ and $3$ respectively, this
force reads,
\begin{eqnarray*}
F_{1} & = & q_{1}\left(v_{12}\times B_{2}^{2}+E_{2}^{2}\right)+q_{1}\left(v_{13}\times B_{3}^{3}+E_{3}^{3}\right)
\end{eqnarray*}
 From Maxwell's invariance theorem, the relation $(A_{3}^{2},V_{3}^{2})=T_{23}(A_{3}^{3},V_{3}^{3})$
is just a recasting of the vector potential $A_{3}^{3}$ in the coordinates
of body $2$. Hence, both potentials take the same value on any given
point of space. Consequently $B_{3}^{2}=B_{3}^{3}$. The transformation
of the electric field and scalar potential reads (also by Theorem
\ref{Theorem-Maxwell's-invariance}), 
\begin{eqnarray*}
E_{3}^{3}=-\nabla V_{3}^{3}-\frac{\partial}{\partial t}A_{3}^{3} & = & -\nabla\left(V_{3}^{2}-v_{32}\cdot A_{3}^{2}\right)-\frac{\partial}{\partial t}A_{3}^{2}-\left(v_{32}\cdot\nabla\right)A_{3}^{2}\\
 & = & v_{32}\times\left(\nabla\times A_{3}^{2}\right)-\nabla V_{3}^{2}-\frac{\partial}{\partial t}A_{3}^{2}\\
 & = & v_{32}\times B_{3}^{2}+E_{3}^{2}
\end{eqnarray*}
Hence,

\begin{eqnarray*}
F_{1} & = & q_{1}\left(v_{12}\times B_{2}^{2}+E_{2}^{2}\right)+q_{1}\left(v_{13}\times B_{3}^{2}+\left(v_{32}\times B_{3}^{2}+E_{3}^{2}\right)\right)\\
 & = & q_{1}\left(v_{12}\times\left(B_{2}^{2}+B_{3}^{2}\right)+\left(E_{2}^{2}+E_{3}^{2}\right)\right)
\end{eqnarray*}
i.e., we have proven
\begin{cor}
Under the assumptions of Theorem \ref{Theorem-Maxwell's-invariance},
the interaction force between electromagnetic bodies is invariant
in front of arbitrary (subjective) Galilean translations of the reference
system.
\end{cor}

In the present framework the electromagnetic force is a Galilean invariant
describing an interaction between primary and secondary circuits.
Changing the reference system for the description while keeping the
relative motion does not affect the force. This is the context in
which translational invariance is proper \citep{sola18b} and where
the Galilean transformation emerges. The situation must be distinguished
from that in which the secondary circuit is put in relative motion
with respect to the primary circuit and, in addition, the reference
frame is changed so that it follows the secondary circuit. Let us
call the primary circuit the \emph{source} and the secondary circuit
the \emph{receiver}. We know that if we put the receiver in relative
motion with respect to the source, it will detect a different signal
(Doppler's effect). Thus, a valid (experimentally accessible) physical
question is: which is the relation between what the receiver perceives
and what it would have perceived had it not been put in relative motion?
However, nothing changes in the relation between source and detector
if we decide to describe it while jogging around the experiment. What
confuses matters is that, in the case of instantaneous action at distance
(not our case), the same Galilean transformation can be used to connect
reference frames of description in (constant) relative motion and
to relate the perceptions of a receiver at rest or in (constant) motion
relative to the source. In short: proper Galilean invariance and Maxwell's
theorem are not in contradiction with experimental facts or with the
theoretical use of Lorentz' transformations for cases of relative
motion.

\section{Discussion}

The attitude towards error defines our science. When our observations
are incompatible with a hypothesis, i.e., they refute it, we can only
think of suppressing and/or replacing the indicted hypothesis as well
as all of its consequences, going back to the point where we had the
wrong idea, and restarting our progress from that point. This has
not been the path historically followed, hence we must ask: which
forces made it impossible? To answer the question we suggest to search
for historical, social and epistemological constrains.

Maxwell worked out his electromagnetism following the path of the
Göttingen school, but he needed to persuade himself of analogies with
matter \citep{maxw56}, hence his introduction of the ether. Lorentz
as well worked out from this mixed epistemological position. More
than a hundred years later we can say that their success comes from
the mathematisation while the problems with their approaches come
from their analogical thoughts, as we have observed for example in
Maxwell's force and Lorentz deduction of his force.\textcolor{brown}{{}
}Modern textbooks have a mixed attitude towards Lorentz force, some
of them pragmatically accepting it as an independent posit (see Introduction).
Even the nature of the velocity in the expression of the force is
not completely clear \citet{assi94} after dropping the ether.

In this work we have shown how to construct the basic elements of
electrodynamics developing a relational electromagnetism that reaches
a higher level of consistency and harmony than the accepted electromagnetism
at the time of Lorentz. Within this approach, the electromagnetic
force corresponding to Lorentz' emerges from the same principles that
produce Maxwell equations, this is, they form a consistent theory
a-priori rather than a-posteriori. The resulting relational electromagnetism
admits a subjective form that complies with the no-arbitrariness-principle
\citep{sola18b}, a principle that is more general and demanding than
Poincaré's principle-of-relativity. The construction begins with one
concept of space and ends within the same concept. We have further
discussed how Galilean invariance takes two different meanings, one
that is proper and derives from the non-arbitrariness-principle and
one that is accidental and restricted to instaneous-action-at-a-distance.
When each matter is given its proper place, there is no contradiction
between Galilean invariance and the use of Lorentz transformations
or experimental results on Doppler' shifts.

Not surprisingly, the elements of our construction can be found in
Gauss, Maxwell, Lorenz, and the Göttingen school. The correspondence
between Lorentz' current and Maxwell's invariance theorem is a key
element for this state of harmony. Lorentz' and Maxwell's transformations,
when restricted to inertial systems reduce to not-so-obvious presentations
of the Galilean group of transformations.

In this relational approach, electromagnetic interactions manifest
themselves in any reference system as waves of speed $C$, refuting
the idea that this is incompatible with classical space-time. But
then, what do we need to drop to accept this theory? The answer is
clear: analogy must be trusted no more that Maxwell did:
\begin{quote}
``...It appears to me, however, that while we derive great advantage
from the recognition of the many analogies between the electric current
and a current of material fluid, we must carefully avoid making any
assumption not warranted by experimental evidence, and that there
is, as yet, no experimental evidence to shew whether the electric
current is really a current of a material substance, or a double current,
or whether its velocity is great or small as measured in feet per
second.'' \citep[{[574],}][]{maxw73}
\end{quote}
The conflict between the Göttingen and Berlin schools illustrates
a radical difference in the conception of science. When physics left
the safe waters of mechanics to penetrate the phenomena that develop
inside matter, the Göttingen school was prepared to resign the matter-space
idealisation, as proposed by Faraday and Lorenz, maintaining the tradition
initiated by Leibniz and Newton that demands science not to introduce
physical hypotheses. The Berlin school struggled to preserve the matter-space
idealisation, introducing metaphysical entities such as the ether
or a body-like carrier of interactions, whose ultimate function is
to facilitate analogical thinking. When we separate action from matter
we need a form to propagate the action of a material entity at distance.
In contrast, when action and matter form a duality as proposed by
Faraday we do not need the ether or action carriers but our intuitive
view of the material world is shaken. While a social decision on this
matter has been taken, the voice of Faraday haunts us from the past:
``we ought to remember that it, in such cases, becomes a prejudice,
and inevitably interferes, more or less, with a clear-sighted judgement''
\citep[p.285]{fara44}.

\bibliographystyle{spbasic}
\bibliography{nuevasreferencias,referencias}

\begin{thebibliography}{39}
\providecommand{\natexlab}[1]{#1}
\providecommand{\url}[1]{{#1}}
\providecommand{\urlprefix}{URL }
\expandafter\ifx\csname urlstyle\endcsname\relax
  \providecommand{\doi}[1]{DOI~\discretionary{}{}{}#1}\else
  \providecommand{\doi}{DOI~\discretionary{}{}{}\begingroup
  \urlstyle{rm}\Url}\fi
\providecommand{\eprint}[2][]{\url{#2}}

\bibitem[{Archibald(1986)}]{arch86}
Archibald T (1986) Carl neumann versus rudolf clausius on the propagation of
  electrodynamic potentials. American Journal of Physics 54(9):786--790

\bibitem[{Arnold(1989)}]{arno89}
Arnold VI (1989) Mathematical Methods of Classical Mechanics. 2nd edition.
  Springer, New York, 1st edition 1978

\bibitem[{Assis(1994)}]{assi94}
Assis A (1994) Weber’s electrodynamics. In: Weber’s Electrodynamics,
  FundaMental theorfes of physics, vol~66, Springer

\bibitem[{Betti(1867)}]{bett67}
Betti E (1867) Sopra la elettrodinamica. Il Nuovo Cimento (1855-1868)
  27(1):402--407

\bibitem[{Clausius(1869)}]{clau69}
Clausius R (1869) Lxii. upon the new conception of electrodynamic phenomena
  suggested by gauss. The London, Edinburgh, and Dublin Philosophical Magazine
  and Journal of Science 37(251):445--456

\bibitem[{D’Agostino(2004)}]{dago04}
D’Agostino S (2004) The bild conception of physical theory: Helmholtz, hertz,
  and schr{\"o}dinger. Physics in Perspective 6(4):372--389

\bibitem[{Einstein(1907)}]{einstein1907relativity}
Einstein A (1907) On the relativity principle and the conclusions drawn from
  it. Jahrb Radioaktivitat Elektronik 4:411--462

\bibitem[{Faraday(1844)}]{fara44}
Faraday M (1844) Experimentl Researches in Electricity (Vol II). Richard Taylor
  and William Francis

\bibitem[{Faraday(1855)}]{fara55}
Faraday M (1855) Experimentl Researches in Electricity (Vol III). Richard and
  John Edward Taylor

\bibitem[{Feynman et~al(1965)Feynman, Leighton, and Sands}]{feynman1965}
Feynman RP, Leighton RB, Sands M (1965) The feynman lectures on physics.
  Addison Wesley Publishing Co

\bibitem[{Gauss(1870)}]{gaus70}
Gauss CF (1870) Carl Friedrich Gauss, Werke, vol~5. K. Gesellschaft der
  Wissenschaften zu Göttingen, digitizing sponsor University of California
  Libraries

\bibitem[{Heaviside(1889)}]{heaviside1889}
Heaviside O (1889) Xxxix. on the electromagnetic effects due to the motion of
  electrification through a dielectric. The London, Edinburgh, and Dublin
  Philosophical Magazine and Journal of Science 27(167):324--339

\bibitem[{Hertz and Walley(1899)}]{hertz89}
Hertz H, Walley JT (1899) The principles of mechanics presented in a new form.
  Macmillian and Company, Limited,
  \urlprefix\url{http://www.archive.org/details/principlesofmech00hertuoft}

\bibitem[{Jackson(1962)}]{jackson1962classical}
Jackson J (1962) Classical Electrodynamics. Wiley, third edition.

\bibitem[{Kirchhoff(1857)}]{kirchhoff1857liv}
Kirchhoff G (1857) Liv. on the motion of electricity in wires. The London,
  Edinburgh, and Dublin Philosophical Magazine and Journal of Science
  13(88):393--412

\bibitem[{Lorentz(1892)}]{lorentz1892CorpsMouvants}
Lorentz HA (1892) La théorie Électromagnétique de maxwell et son application
  aux corps mouvants. Archives Néerlandaises des Sciences exactes et
  naturelles XXV:363--551,
  \urlprefix\url{https://www.biodiversitylibrary.org/item/181480#page/417/mode/1up},
  scanned by Biodiversity Heritage Library from holding at Harvard University
  Botany Libraries

\bibitem[{Lorentz(1895)}]{lorentz1895attempt}
Lorentz HA (1895) Attempt of a theory of electrical and optical phenomena in
  moving bodies. Leiden: EJ Brill, Leiden
  \urlprefix\url{https://en.wikisource.org/wiki/Translation:Attempt_of_a_Theory_of_Electrical_and_Optical_Phenomena_in_Moving_Bodies}

\bibitem[{Lorentz(1899)}]{lorentz1899simplified}
Lorentz HA (1899) Simplified theory of electrical and optical phenomena in
  moving systems. Koninklijke Nederlandse Akademie van Wetenschappen
  Proceedings Series B Physical Sciences 1:427--442

\bibitem[{Lorenz(1861)}]{lore61}
Lorenz L (1861) Xlix. on the determination of the direction of the vibrations
  of polarized light by means of diffraction. The London, Edinburgh, and Dublin
  Philosophical Magazine and Journal of Science 21(141):321--331

\bibitem[{Lorenz(1867)}]{lore67}
Lorenz L (1867) Xxxviii. on the identity of the vibrations of light with
  electrical currents. The London, Edinburgh, and Dublin Philosophical Magazine
  and Journal of Science 34(230):287--301

\bibitem[{Maxwell(1856)}]{maxw56}
Maxwell JC (1856) On faraday's lines of force. Transactions of the Cambridge
  Philosophical Society X (Part 1):155--229

\bibitem[{Maxwell(1865)}]{maxw64}
Maxwell JC (1865) A dynamical theory of the electromagnetic field. Proceedings
  of the Royal Society (United Kingdom)

\bibitem[{Maxwell(1873)}]{maxw73}
Maxwell JC (1873) A Treatise on Electricity and Magnetism, vol 1 and 2. Dover
  (1954)

\bibitem[{Natiello and Solari(2019)}]{NatielloManuscript-NATTCO-21}
Natiello M, Solari HG (2019) The consruction of electromagnetism,
  \urlprefix\url{https://philpapers.org/rec/NATTCO-21}

\bibitem[{Neumann(1868)}]{neum68}
Neumann C (1868) Die principien der elektrodynamik. Eine mathematische
  Untersuchung Verlag der Lauppschen Buchhandlung, T{\"u}bingen Reprinted in
  Mathematischen Annalen, Vol. 17, pp. 400 - 434 (1880).

\bibitem[{Neumann(1846)}]{neum1846-induction}
Neumann F (1846) Allgemeine gesetze der inducierten elektrischen strome. pogg.
  Annalen der Physik (Poggendorf) 143(1):31--44

\bibitem[{Panofsky and Phillips(1955)}]{panofsky1955classical}
Panofsky W, Phillips M (1955) Classical Electricity and Magnetism.
  Addison-Wesley, Reading (Mass.), London

\bibitem[{Poggendorff(1857)}]{poge57}
Poggendorff JC (1857) Comment on the paper by prof. kirchhoff. Annalen der
  Physik 100:351--352, \urlprefix\url{http://arxiv.org/abs/1912.05930}, english
  translattion by A. Assis

\bibitem[{Purcell(1965)}]{purc65}
Purcell EM (1965) Berkeley Physics Course. Mc Graw-Hill, New York, v 2,
  Electricity and Magnetism

\bibitem[{Riemann(1867)}]{riem67}
Riemann B (1867) Xlvii. a contribution to electrodynamics. The London,
  Edinburgh, and Dublin Philosophical Magazine and Journal of Science
  34(231):368–--372, translated from Poggendorff's Annalen, No. 6, 1867. Laid
  before the Royal Society of Sciences at Göttingen on the 10th of February
  1858. Published posthumous.

\bibitem[{Ritz(1908)}]{ritz1908recherches}
Ritz W (1908) Recherches critiques sur l'{\'e}lectrodynamique g{\'e}n{\'e}rale.
  Helvetica Physica Acta 13

\bibitem[{Rohrlich(2007)}]{rohr07classical}
Rohrlich F (2007) Classical charged particles: foundations of their theory.
  World Scientific, first published 1965 by Westview Press

\bibitem[{Solari and Natiello(2021)}]{Solari21-phenomenologico}
Solari HG, Natiello M (2021) Science, dualities and the fenomenological map,
  \urlprefix\url{https://philpapers.org/rec/SOLSDA-4}

\bibitem[{Solari and Natiello(2018)}]{sola18b}
Solari HG, Natiello MA (2018) A constructivist view of newton’s mechanics.
  Foundations of Science 24:307,
  \urlprefix\url{https://doi.org/10.1007/s10699-018-9573-z}

\bibitem[{Soper(1975)}]{sope75classical}
Soper DE (1975) Classical field theory. John Wiley \& Sons

\bibitem[{Thomson(1881)}]{thomson1881}
Thomson J (1881) On the electric and magnetic effects produced by the motion of
  electrified bodies. Philosophical Magazine 11:229--249

\bibitem[{Weber(1846)}]{webe46}
Weber W (1846) Determinations of electromagnetic meassure; concerning a
  universal law of electrical action. Prince Jablonowski Society (Leipzig) pp
  211--378, translated by Susan P. Johnson and edited by Laurence Hecht and A.
  K. T. Assis from Wilhelm Weber, Elektrodynamische Maassbestimmungen: Ueber
  ein allgemeines Grundgesetz der elektrischen Wirkung, Werke, Vol. III:
  Galvanismus und Electrodynamik, part 1, edited by H. Weber (Berlin: Julius
  Springer Verlag, 1893), pp. 25-214

\bibitem[{Weber(1864)}]{webe64}
Weber W (1864) Elektrodynamische maassbestimmungen insbesondere über
  elektrische schwingungen. Abhandlungen der Königl Sächs Geselschaft der
  Wissenschaften, mathematisch-physische Klasse Reprinted in Wilhelm Weber’s
  Werke, Vol. 4, H. Weber (ed.), (Springer, Berlin, 1894), pp. 105-241

\bibitem[{Zangwill(2012)}]{zang12modern}
Zangwill A (2012) Modern electrodynamics. Cambridge University Press

\end{thebibliography}

\appendix

\section{Lemmas and Proofs\label{sec:Proofs}}
\begin{lem}
\label{lem:wave}${\displaystyle A(x,t)=\frac{\mu_{0}}{4\pi}\int_{U}\left(\frac{j(y,t-\frac{1}{C}|x-y|)}{|x-y|}\right)\,d^{3}y\Rightarrow\Box A=-\mu_{0}j}$,
and similarly for $\epsilon_{0}\Box V=-\rho$, where $\Box\equiv{\displaystyle \Delta-\frac{1}{C^{2}}\frac{\partial^{2}}{\partial t^{2}}}$.
\end{lem}

\begin{proof}
We perform the calculation in detail only for $A$, since the other
one is similar. We use the shorthand $r=|x-y|$.

\begin{eqnarray*}
\nabla_{x}A_{i} & = & \frac{\mu_{0}}{4\pi}\int d^{3}y\,\left(j_{i}\nabla_{x}\frac{1}{r}-\frac{\frac{\partial}{\partial t}j_{i}\nabla_{x}\frac{r}{C}}{r}\right)\\
\Delta A_{i} & = & \nabla_{x}\cdot\nabla_{x}A_{i}\\
 & = & \frac{\mu_{0}}{4\pi}\int d^{3}y\,\left(j_{i}\Delta\frac{1}{r}-2\left(\nabla_{x}\frac{1}{r}\right)\cdot\left(\frac{\partial}{\partial t}j_{i}\nabla_{x}\frac{r}{C}\right)-\frac{\frac{\partial}{\partial t}j_{i}\Delta\frac{r}{C}}{r}+\frac{\frac{\partial^{2}}{\partial t^{2}}j_{i}}{r}|\nabla_{x}\frac{r}{C}|^{2}\right)
\end{eqnarray*}
Moreover, standard vector calculus identities give
\begin{eqnarray*}
\frac{\partial}{\partial t}j_{i}\left(2\nabla\frac{1}{r}\cdot\nabla\frac{r}{C}+\frac{\Delta\frac{r}{C}}{r}\right) & = & 0\\
|\nabla\frac{r}{C}|^{2} & = & \frac{1}{C^{2}}
\end{eqnarray*}
and therefore
\begin{eqnarray*}
\Delta A_{i}(x,t) & = & \frac{\mu_{0}}{4\pi}\int d^{3}y\,j_{i}(y,t-\frac{r}{C})\Delta\left(\frac{1}{r}\right)+\left(\frac{1}{C^{2}}\right)\frac{\mu_{0}}{4\pi}\int d^{3}y\,\frac{\partial^{2}}{\partial t^{2}}\frac{j_{i}(y,t-\frac{r}{C})}{r}
\end{eqnarray*}
The time derivative in the last term can be extracted outside the
integral, thus yielding,
\begin{eqnarray*}
\Box A_{i}(x,t) & = & \Delta A_{i}(x,t)-\left(\frac{1}{C^{2}}\right)\frac{\mu_{0}}{4\pi}\int d^{3}y\,\frac{\partial^{2}}{\partial t^{2}}\frac{j_{i}(y,t-\frac{r}{C})}{r}\\
 & = & \Delta A_{i}(x,t)-\left(\frac{1}{C^{2}}\right)\frac{\partial^{2}}{\partial t^{2}}A_{i}(x,t)\\
 & = & \frac{\mu_{0}}{4\pi}\int d^{3}y\,j_{i}(y,t-\frac{|x-y|}{C})\Delta\left(\frac{1}{r}\right)\\
 & = & -\mu_{0}j_{i}(x,t)
\end{eqnarray*}
\end{proof}
\begin{lem}
\label{lem:GaugeLorenz} The continuity equation along with eq.\ref{eq:Lorenz}
imply ${\displaystyle \nabla\cdot A+{\displaystyle \frac{1}{C^{2}}\frac{\partial V}{\partial t}}=0}$
(the Lorenz gauge).
\end{lem}

\begin{proof}
Still using the shorthand $r=|x-y|$, and applying Gauss theorem over
volume integrals of total divergences of functions vanishing sufficiently
fast at infinity,
\begin{eqnarray*}
0 & = & \nabla\cdot A+\frac{1}{C^{2}}\frac{\partial V}{\partial t}=\\
 & = & \int d^{3}y\left(\nabla_{x}\cdot\frac{j(y,t-\frac{r}{C})}{|x-y|}+{\displaystyle \frac{{\displaystyle \frac{\partial}{\partial t}\rho(y,t-\frac{r}{C})}}{|x-y|}}\right)\\
 & = & \int d^{3}y\left(\nabla_{x}\frac{1}{|x-y|}\cdot j(y,t-\frac{r}{C})-\frac{1}{|x-y|}\left[\frac{\partial}{\partial s}j(y,s)\right]_{s=t-\frac{r}{C}}\cdot\nabla_{x}\frac{r}{C}+{\displaystyle \frac{{\displaystyle \frac{\partial}{\partial t}\rho(y,t-\frac{r}{C})}}{|x-y|}}\right)\\
 & = & \int d^{3}y\left(-\nabla_{y}\frac{1}{|x-y|}\cdot j(y,t-\frac{r}{C})-\frac{1}{|x-y|}\left[\frac{\partial}{\partial s}j(y,s)\right]_{s=t-\frac{r}{C}}\cdot\nabla_{x}\frac{r}{C}+{\displaystyle \frac{{\displaystyle \frac{\partial}{\partial t}\rho(y,t-\frac{r}{C})}}{|x-y|}}\right)\\
 & = & \int d^{3}y\left(\frac{\nabla_{y}\cdot j(y,t-\frac{r}{C})}{|x-y|}-\frac{1}{|x-y|}\left[\frac{\partial}{\partial s}j(y,s)\right]_{s=t-\frac{r}{C}}\cdot\nabla_{x}\frac{r}{C}+{\displaystyle \frac{{\displaystyle \frac{\partial}{\partial t}\rho(y,t-\frac{r}{C})}}{|x-y|}}\right)\\
 & = & \int d^{3}y\left(\frac{\nabla_{y}\cdot\left[j(y,s)\right]_{s=t-\frac{r}{C}}}{|x-y|}-\frac{1}{|x-y|}\left[\frac{\partial}{\partial s}j(y,s)\right]_{s=t-\frac{r}{C}}\cdot\left(\nabla_{y}+\nabla_{x}\right)\frac{r}{C}+{\displaystyle \frac{{\displaystyle \frac{\partial}{\partial t}\rho(y,t-\frac{r}{C})}}{|x-y|}}\right)\\
 & = & \int d^{3}y\left(\frac{\left[\frac{\partial}{\partial s}\rho(y,s)+\nabla_{y}\cdot j(y,s)\right]_{s=t-\frac{r}{C}}}{|x-y|}\right)
\end{eqnarray*}
since $\left(\nabla_{y}+\nabla_{x}\right)|x-y|=0$.
\end{proof}
\begin{lem}
\label{lem:eq4-no-ether} Lemma \ref{lem:GaugeLorenz}, together with
eqs. \ref{eq:magneticfield}, \ref{eq:electricfield} and \ref{eq:Lorenz}
imply eq. \ref{eq:displacement}.
\end{lem}

\begin{proof}
The proof requires standard applications of vector calculus. From
eq. \ref{eq:Lorenz} we derive
\begin{eqnarray*}
\mu_{0}j & = & -\left(\nabla\left(\nabla\cdot A\right)-\nabla\times\left(\nabla\times A\right)-\frac{1}{C^{2}}\frac{\partial^{2}A}{\partial t^{2}}\right)\\
 & = & \frac{1}{C^{2}}\left(\nabla\frac{\partial V}{\partial t}+\frac{\partial^{2}A}{\partial t^{2}}\right)+\nabla\times B\\
 & = & -\frac{1}{C^{2}}\frac{\partial E}{\partial t}+\nabla\times B
\end{eqnarray*}
\end{proof}
\begin{lem}
\label{lem:energy-final} Eq. \ref{eq:power} is equivalent to eq.
\ref{eq:energy} (the total electromagnetic energy), up to the volume
integral of the gradient of a function that vanishes at infinity.
\end{lem}

\begin{proof}
Under the general assumption that $\int\nabla\cdot F(x,t)\,d^{3}x$
vanishes at infinity, being $F$ a vector function that decays sufficiently
fast (i.e., faster than $r^{-2}$), we obtain
\begin{eqnarray*}
 &  & \int_{0}^{t}dt\int d^{3}x\left(\left(\nabla V+\frac{\partial A}{\partial t}\right)\cdot j\right)=\\
 & = & \int_{0}^{t}dt\int d^{3}x\left(\nabla\cdot(Vj)-V\nabla\cdot j+\frac{\partial A}{\partial t}\cdot\left(-\epsilon_{0}\frac{\partial E}{\partial t}+\frac{1}{\mu_{0}}\nabla\times B\right)\right)\\
 & = & \int_{0}^{t}dt\int d^{3}x\left(V\frac{\partial\rho}{\partial t}+\frac{\partial A}{\partial t}\cdot\left(-\epsilon_{0}\frac{\partial E}{\partial t}+\frac{1}{\mu_{0}}\nabla\times B\right)\right)\\
 & = & \int_{0}^{t}dt\int d^{3}x\left(\epsilon_{0}V\nabla\cdot\frac{\partial E}{\partial t}-\epsilon_{0}\frac{\partial A}{\partial t}\cdot\frac{\partial E}{\partial t}+\frac{1}{\mu_{0}}\frac{\partial A}{\partial t}\cdot\nabla\times B\right)\\
 & = & \int_{0}^{t}dt\int d^{3}x\left(\epsilon_{0}\nabla\cdot(V\frac{\partial E}{\partial t})-\epsilon_{0}\nabla V\cdot\frac{\partial E}{\partial t}-\epsilon_{0}\frac{\partial A}{\partial t}\cdot\frac{\partial E}{\partial t}+\frac{1}{\mu_{0}}\left(\frac{\partial B}{\partial t}\cdot B\right)-\frac{1}{\mu_{0}}\nabla\cdot\left(\frac{\partial A}{\partial t}\times B\right)\right)\\
 & = & \int_{0}^{t}dt\int d^{3}x\left(\epsilon_{0}E\cdot\frac{\partial E}{\partial t}+\frac{1}{\mu_{0}}\left(\frac{\partial B}{\partial t}\cdot B\right)\right)\\
 & = & \frac{1}{2}\int_{0}^{t}dt\frac{\partial}{\partial t}\int d^{3}x\left(\epsilon_{0}|E|^{2}+\frac{1}{\mu_{0}}|B|^{2}\right)\\
 & = & \frac{1}{2}\int d^{3}x\left(\epsilon_{0}|E|^{2}+\frac{1}{\mu_{0}}|B|^{2}\right)
\end{eqnarray*}
\end{proof}
where all integrals involving total divergences have been set to zero
by Gauss' theorem. We have also used the continuity equation.
\begin{lem}
\label{lem:action-final}Up to an overall function of time and the
divergence of a function vanishing sufficiently fast at infinity,
the electromagnetic action satisfies
\begin{eqnarray*}
{\cal {\cal A}} & = & \frac{1}{2}\int dt\int\left(\frac{1}{\mu_{0}}|B|^{2}-\epsilon_{0}|E|^{2}\right)\,d^{3}x\\
 & = & \frac{1}{2}\int dt\,\int\left(A\cdot j-\rho V\right)\,d^{3}x
\end{eqnarray*}
\end{lem}

\begin{proof}
The proof requires standard vector calculus operations on the integrand,
namely 
\begin{eqnarray*}
A\cdot j-\rho V & = & A\cdot\left(\frac{1}{\mu_{0}}\nabla\times B-\epsilon_{0}\frac{\partial E}{\partial t}\right)-\epsilon_{0}V\nabla\cdot E\\
 & = & \frac{1}{\mu_{0}}\left(|B|^{2}-\nabla\cdot\left(A\times B\right)\right)-\epsilon_{0}A\cdot\frac{\partial E}{\partial t}-\epsilon_{0}\nabla\cdot(VE)+\epsilon_{0}E\cdot\nabla V\\
 & = & \frac{1}{\mu_{0}}\left(|B|^{2}-\nabla\cdot\left(A\times B\right)\right)-\epsilon_{0}A\cdot\frac{\partial E}{\partial t}-\epsilon_{0}\nabla\cdot(VE)+\epsilon_{0}E\cdot\left(-E-\frac{\partial A}{\partial t}\right)\\
 & = & \frac{1}{\mu_{0}}|B|^{2}-\epsilon_{0}|E|^{2}-\nabla\cdot\left(\frac{1}{\mu_{0}}A\times B+\epsilon_{0}VE\right)-\epsilon_{0}\frac{\partial}{\partial t}\left(A\cdot E\right)
\end{eqnarray*}
\end{proof}
\begin{lem}
\label{lem:gauge-independent} The result of Lemma \ref{lem:action-final}
is independent of the choice of gauge.
\end{lem}

\begin{proof}
Modifying the potentials with a sufficiently smooth function $\Lambda(x,t)$
also vanishing adequately at infinity and such that $A^{\prime}=A+\nabla\Lambda$
and $V^{\prime}=V-{\displaystyle \frac{\partial\Lambda}{\partial t}}$,
we obtain
\begin{eqnarray*}
A^{\prime}\cdot j-\rho V^{\prime} & = & A\cdot j+(\nabla\Lambda\cdot j)-\rho(V-\frac{\partial\Lambda}{\partial t})\\
 & = & A\cdot j-\rho V+\left(\nabla\Lambda\cdot j+\rho\frac{\partial\Lambda}{\partial t}\right)\\
 & = & A\cdot j-\rho V+\left(\nabla\cdot(\Lambda j)-\Lambda\nabla\cdot j+\rho\frac{\partial\Lambda}{\partial t}\right)\\
 & = & A\cdot j-\rho V+\nabla\cdot(\Lambda j)-\frac{\partial}{\partial t}(\rho\Lambda)
\end{eqnarray*}
\end{proof}

\subsubsection{Proof of Theorem \ref{thm:TheForce}:\label{subsec:TheProof}}
\begin{proof}
For the sake of convenience, we may express the action using Lemma
\ref{lem:action-final} before proceeding:

\begin{equation}
\delta_{\bar{x}(t)}{\cal A}=\delta_{\bar{x}(t)}\int dt\,\int d^{3}x\,\left[A\cdot j-V\rho\right]=\int dt\,\int d^{3}x\,[\delta A\cdot j+A\cdot\delta j-\delta V\rho-V\delta\rho]=0\label{eq:variation}
\end{equation}
Potentials charge and current can be split in primary and secondary
circuits, namely $A=A_{1}+A_{2}$ (and similarly for $V,j,\rho)$
and the variation can be formulated as moving the secondary circuit
with respect to a fixed primary circuit. Only the interacting terms
actually vary, so the integrand reads $\left(\delta A_{2}\cdot j_{1}-\delta V_{2}\rho_{1}\right)+\left(A_{1}\cdot\delta j_{2}-V_{1}\delta\rho_{2}\right)$.
We transform the first parenthesis using Maxwell equations, Lemma
\ref{lem:wave}: 
\begin{eqnarray*}
\int dt\int d^{3}x\,\left[\delta A_{2}\cdot j_{1}-\delta V_{2}\rho_{1}\right] & = & \int dt\int d^{3}x\,\left[-\mu_{0}\delta A_{2}\cdot\square A_{1}+\epsilon_{0}\delta V_{2}\square V_{1}\right]
\end{eqnarray*}
The particular variation reflecting the problem (moving subsystem
2 with respect to a fixed subsystem 1) and partial integrations in
space and time allow to transform
\[
\int dt\int d^{3}x\,\left[-\mu_{0}\delta A_{2}\cdot\square A_{1}+\epsilon_{0}\delta V_{2}\square V_{1}\right]=\int dt\int d^{3}x\,\left[-\mu_{0}\delta A_{1}\cdot\square A_{2}+\epsilon_{0}\delta V_{1}\square V_{2}\right]
\]
up to an overall divergence whose contribution vanishes by Gauss theorem
and an overall function of time that do not contribute to the variation.
Hence, the required variation reads (dropping the index $\bar{x}(t)$),
\[
\delta{\cal A}=\int dt\int d^{3}x\,\left[A_{1}\cdot\delta j_{2}-V_{1}\delta\rho_{2}\right]
\]

The variation of the current and of the charge distribution due to
the motion of the probe relative to the primary circuit are, after
eq. \ref{eq:Current},

\begin{eqnarray*}
\delta\bar{j}_{2} & = & -(\delta\bar{x}\cdot\nabla)\bar{j}_{2}+\bar{\rho}_{2}\delta\dot{\bar{x}}\\
\delta\bar{\rho}_{2} & = & -(\delta\bar{x}\cdot\nabla)\bar{\rho}_{2}
\end{eqnarray*}
We have then 
\begin{eqnarray}
\delta{\cal A} & = & \delta\int dt\,{\cal L}=\int dt\int d^{3}x\,\left[A_{1}\cdot\delta\bar{j}_{2}-V_{1}\delta\bar{\rho}_{2}\right]\nonumber \\
 & = & \int dt\,\int d^{3}x\,\left[A_{1}\cdot\left(-(\delta\bar{x}\cdot\nabla)\bar{j}_{2}+\bar{\rho}_{2}\delta\dot{\bar{x}}\right)-V_{1}\left(-(\delta\bar{x}\cdot\nabla)\bar{\rho}_{2}\right)\right]\label{eq:variada}\\
 & = & \int dt\,\int d^{3}x\,\left[\bar{j}_{2}\cdot\left(\delta\bar{x}\cdot\nabla\right)A_{1}-\bar{\rho}_{2}\left(\delta\bar{x}\cdot\nabla\right)V_{1}-\delta\bar{x}\cdot\frac{\partial}{\partial t}\left(A_{1}\bar{\rho}_{2}\right)\right]\nonumber 
\end{eqnarray}
The last line is obtained after some integrations by parts and following
the cancellation of a whole divergence via Gauss theorem and of an
overall function of time via the variational constraints. In particular:
\begin{eqnarray*}
\left(\left(\delta\bar{x}\cdot\nabla\right)\bar{j}_{2}\right)\cdot A_{1} & = & \left(\delta\bar{x}\cdot\nabla\right)\left(\bar{j}_{2}\cdot A_{1}\right)-\bar{j}_{2}\cdot\left(\delta\bar{x}\cdot\nabla\right)A_{1}\\
 & = & \nabla\cdot\left(\delta\bar{x}\left(\bar{j}_{2}\cdot A_{1}\right)\right)-\bar{j}_{2}\cdot\left(\delta\bar{x}\cdot\nabla\right)A_{1}\\
V_{1}(\delta\bar{x}\cdot\nabla)\bar{\rho}_{2} & = & \delta\bar{x}\cdot\nabla\left(V_{1}\bar{\rho}_{2}\right)-\bar{\rho}_{2}\left(\delta\bar{x}\cdot\nabla\right)V_{1}\\
 & = & \nabla\cdot\left(\delta\bar{x}V_{1}\bar{\rho}_{2}\right)-\bar{\rho}_{2}\left(\delta\bar{x}\cdot\nabla\right)V_{1}\\
A_{1}\cdot\left(\bar{\rho}_{2}\delta\dot{\bar{x}}\right) & = & \frac{\partial}{\partial t}\left(\rho_{2}A_{1}\cdot\delta\bar{x}\right)-\delta\bar{x}\cdot\frac{\partial}{\partial t}\left(A_{1}\bar{\rho}_{2}\right)
\end{eqnarray*}
Further transformation with mathematical identities allows us to write
\begin{eqnarray*}
\bar{j}_{2}\cdot\left(\delta\bar{x}\cdot\nabla\right)A_{1}-\delta\bar{x}\cdot A_{1}\frac{\partial}{\partial t}\bar{\rho}_{2} & = & \bar{j}_{2}\cdot\left(\delta\bar{x}\cdot\nabla\right)A_{1}+\left(\delta\bar{x}\cdot A_{1}\right)\left(\nabla\cdot\bar{j}_{2}\right)\\
 & = & \bar{j}_{2}\cdot\left(\delta\bar{x}\cdot\nabla\right)A_{1}+\nabla\cdot\left(\bar{j}_{2}\left(\delta\bar{x}\cdot A_{1}\right)\right)-\bar{j}_{2}\cdot\nabla\left(\delta\bar{x}\cdot A_{1}\right)\\
 & = & \nabla\cdot\left(\bar{j}_{2}\left(\delta\bar{x}\cdot A_{1}\right)\right)-\bar{j}_{2}\cdot\delta\bar{x}\times\left(\nabla\times A_{1}\right)\\
 & = & \nabla\cdot\left(\bar{j}_{2}\left(\delta\bar{x}\cdot A_{1}\right)\right)+\delta\bar{x}\cdot\bar{j}_{2}\times\left(\nabla\times A_{1}\right)
\end{eqnarray*}
and finally, after applying Gauss theorem again,

\[
\int dt\int d^{3}x\,\delta\bar{x}\cdot\left[\bar{j}_{2}\times B_{1}+\bar{\rho}_{2}E_{1}\right].
\]
This is, following the standard use of Hamilton's principle in mechanics
we arrive to an electromagnetic contribution to the force on the probe
\[
F_{em}=\bar{j}_{2}\times B_{1}+\bar{\rho}_{2}E_{1}
\]
\end{proof}

\end{document}